\documentclass[letterpaper, 10 pt, conference]{ieeeconf}
\usepackage{amsfonts}
\usepackage{amsmath}
\usepackage{amssymb}
\usepackage{color}
\usepackage{cite}
\usepackage{epsfig}
\usepackage{graphics}
\usepackage{graphicx}
\usepackage{ifpdf}
\usepackage{mathptmx}
\usepackage{subfigure}
\usepackage{times}
\usepackage{times}

\IEEEoverridecommandlockouts                              
\newtheorem{theorem}{Theorem}

\newtheorem{lemma}[theorem]{Lemma}

\newtheorem{remark}[theorem]{Remark}

\DeclareMathOperator{\var}{var}
\DeclareMathOperator{\vect}{vec}
\DeclareMathOperator{\cov}{cov}
\DeclareMathOperator{\Ber}{Ber}

\begin{document}

\title{{\LARGE \textbf{On Asymptotic Consensus Value in Directed Random
Networks}}}
\author{Victor M. Preciado, Alireza Tahbaz-Salehi, and Ali Jadbabaie \thanks{%
This research is supported in parts by the following grants: DARPA/DSO
SToMP, NSF ECS-0347285, ONR MURI N000140810747, and AFOSR: Complex Networks
Program.} \thanks{%
Victor Preciado is with Department of Electrical and Systems Engineering,
University of Pennsylvania, Philadelphia, PA 19104 (e-mail:
preciado@seas.upenn.edu)} \thanks{%
Alireza Tahbaz-Salehi is with Department of Electrical Engineering and
Computer Science, Massachusetts Institute of Technology, Cambridge, MA 02139
(e-mail: alirezat@mit.edu).} \thanks{%
Ali Jadbabaie is with the General Robotics, Automation, Sensing and
Perception (GRASP) Laboratory, Department of Electrical and Systems
Engineering, University of Pennsylvania, Philadelphia, PA 19104 (e-mail:
jadbabai@seas.upenn.edu).}}
\maketitle

\begin{abstract}
We study the asymptotic properties of distributed consensus algorithms over
switching directed random networks. More specifically, we focus on consensus
algorithms over independent and identically distributed, directed random
graphs, where each agent can communicate with any other agent with some
exogenously specified probability. While different aspects of consensus
algorithms over random switching networks have been widely studied, a
complete characterization of the distribution of the asymptotic value for
general \textit{asymmetric} random consensus algorithms remains an open
problem. In this paper, we derive closed-form expressions for the mean and
an upper bound for the variance of the asymptotic consensus value, when the
underlying network evolves according to an i.i.d. \textit{directed} random
graph process. We also provide numerical simulations that illustrate our
results.
\end{abstract}

\thispagestyle{empty} \pagestyle{empty}





\section{Introduction}

Distributed consensus algorithms have attracted a significant amount of
attention in the past few years. Besides their wide range of applications in
distributed and parallel computation \cite{TsitsiklisThesis}, distributed
control \cite{Jad2003},\cite{Olfati2007} and robotics \cite{Bullo2005}, they
have also been used as models of opinion dynamics and belief formation in
social networks \cite{DeGroot1974, Jackson_Naive}. The central focus in this
vast body of literature is to study whether a group of agents in a network,
with \textit{local} communication capabilities can reach a \textit{global}
agreement, using simple, deterministic information exchange protocols.%
\footnote{%
For a survey on the most recent works in this area see \cite{Olfati2007}.}

In recent years, there has also been some interest in understanding the
behavior of consensus algorithms in random settings \cite%
{Hatano2004,WuRandomConsensus,Porfiri_TAC,Tahbaz_Jad_TAC, Picci_Taylor_CDC,
Asu_misinformation}. The randomness can be either due to the choice of a
randomized network communication protocol or simply caused by the potential
unpredictability of the environment in which the distributed consensus
algorithm is implemented \cite{Zampieri}. It is recently shown that
consensus algorithms over i.i.d. random networks lead to a global agreement
on a possibly random value, as long as the network is connected in
expectation \cite{Tahbaz_Jad_TAC}. While different aspects of consensus
algorithms over random switching networks, such as conditions for
convergence \cite{Hatano2004, WuRandomConsensus, Porfiri_TAC, Tahbaz_Jad_TAC}
and the speed of convergence \cite{Zampieri}, have been widely studied, a
characterization of the distribution of the asymptotic consensus value has
attracted little attention. Two notable exceptions are Boyd \textit{et al.} 
\cite{Boyd2005b}, who study the asymptotic behavior of the random consensus
value in the special case of symmetric networks, and Tahbaz-Salehi and
Jadbabaie \cite{Tahbaz_Jad_TAC_2}, who compute the mean and variance of the
consensus value for general i.i.d. graph processes. Nevertheless, a complete
characterization of the distribution of the asymptotic value for general 
\textit{asymmetric} random consensus algorithms remains an open problem.

In this paper, we study asymptotic properties of consensus algorithms over a
general class of switching, directed random graphs. More specifically,
building on the results of \cite{Tahbaz_Jad_TAC_2}, we derive closed-form
expressions for the mean and an upper-bound for the variance of the
asymptotic consensus value, when the underlying network evolves according to
an i.i.d. \textit{directed} random graph process. In our model, at each time
period, a directed communication link is established between two agents with
some exogenously specified probability. Due to the potential asymmetry in
pairwise communications between different agents, the asymptotic value of
consensus is not guaranteed to be the average of the initial conditions.
Instead, agents will asymptotically agree on some random value in the convex
hull of the initial conditions. Furthermore, our closed-form
characterization of the variance provides a quantitative measure of how
dispersed the random agreement point is around the average of the initial
conditions in terms of the fundamentals of the model, namely, the structure
of the network, the exogenous probabilities of communication, and the
initial conditions.

The rest of the paper is organized as follows. In the next section, we
describe our model of random consensus algorithms. 
In Sections \ref{Section Mean Analysis} and \ref{Section_Quad_Analysis}, we
derive an explicit expression for the mean and an upper bound for the
variance of the limiting consensus value over switching directed random
graphs, respectively. Section V contains simulations of our results and
Section VI concludes the paper.


\section{Consensus Over Switching Directed Random Graphs}

\label{Section_Model}

Consider the discrete-time linear dynamical system 
\begin{equation}
\mathbf{x}\left( k\right) =W_{k}\mathbf{x}\left( k-1\right) ,
\label{consensus_update}
\end{equation}%
where $k\in \left\{ 1,2,\dots \right\} $ is the discrete time index, $%
\mathbf{x}(k)\in \mathbb{R}^{n}$ is the state vector at time $k$, and $%
\{W_{k}\}_{k=1}^{\infty }$ is a sequence of stochastic matrices. We
interpret (\ref{consensus_update}) as a distributed scheme where a
collection of agents, labeled 1 through $n$, update their state values as a
convex combination of the state values of their neighbors at the previous
time step. Given this interpretation, $\mathbf{x}_{i}(k)$ corresponds to the
state value of agent $i$ at time $k$, and $W_{k}$ captures the neighborhood
relation between different agents at time $k$: the $ij$ element of $W_{k}$
is positive only if agent $i$ has access to the state of agent $j$. For the
remainder of the paper, we assume that the weight matrices $W_{k}$ are
randomly generated by an independent and identically distributed matrix
process.

We say that the dynamical system (\ref{consensus_update}) reaches \textit{%
consensus} asymptotically on some path $\{W_{k}\}_{k=1}^{\infty }$, if along
that path, there exists $x^{\ast }\in \mathbb{R}$ such that $%
x_{i}(k)\rightarrow x^{\ast }$ for all $i$ as $k\rightarrow \infty $. We
refer to $x^{\ast }$ as the \emph{asymptotic consensus value}. It is
well-known that for i.i.d. random networks, the dynamical system (\ref%
{consensus_update}) reaches consensus on almost all paths if and only if the
graph corresponding to the communications between agents is connected in
expectation. More precisely, it was shown in \cite{Tahbaz_Jad_TAC} that $%
W_{k}\dots W_{2}W_{1}\longrightarrow \mathbf{1}d^{T}$ almost surely $-$
where $d$ is some random vector $-$ if and only if the second largest
eigenvalue modulus of $\mathbb{E}W_{k}$ is subunit. Clearly, under such
conditions, the dynamical system in (\ref{consensus_update}) reaches
consensus with probability one where the consensus value is a random
variable equal to $x^{\ast }=d^{T}\mathbf{x}(0)$, where $\mathbf{x}(0)$ is
the vector of initial conditions.

A complete characterization of the random consensus value $x^{\ast }$ is an
open problem. However, it is possible to compute its mean and variance in
terms of the first two moments of the i.i.d. weight matrix process. In \cite%
{Tahbaz_Jad_TAC_2}, the authors prove that the conditional mean of the
random consensus value is given by the random consensus value are given by 
\begin{equation}
\mathbb{E}x^{\ast }=\mathbf{x}(0)^{T}\mathbf{v}_{1}(\mathbb{E}W_{k}),
\label{mean expression}
\end{equation}%
and its conditional variance is equal to 
\begin{eqnarray}
\lefteqn{\var(x^{\ast })=\left[ \mathbf{x}(0)\otimes \mathbf{x}(0)\right]
^{T}\vect(\cov(d))}  \label{variance_expression} \\
&&\!\!\!\!\!=\left[ \mathbf{x}(0)\otimes \mathbf{x}(0)\right] ^{T}\mathbf{v}%
_{1}(\mathbb{E}\left[ W_{k}\otimes W_{k}\right] )-[\mathbf{x}(0)^{T}\mathbf{v%
}_{1}(\mathbb{E}W_{k})]^{2}  \notag
\end{eqnarray}%
where $\mathbf{v}_{1}\left( \cdot \right) $ denotes the normalized left
eigenvector corresponding to the unit eigenvalue, and $\otimes $ denotes the
Kronecker product. In \cite{PTJ10}, the authors derive an explicit
expression for the mean and variance of $x^{\ast }$ in the particular case
of a switching Erd\"{o}s-R\'{e}nyi random graph process. In the following,
we shall use (\ref{variance_expression}) to extend these results and derive
an explicit expression for the mean and an upper-bound for the variance of
the asymptotic consensus value over a wider class of switching, directed
random graph processes.




\section{Mean Analysis for Directed Random Graphs Processes}

\label{Section Mean Analysis}

\subsection{Directed Random Graph Process}

Consider a connected undirected graph $G_{c}=\left( V,E_{c}\right) $ with a
fixed set of vertices $V=\left[ n\right] ,$ and unweighted edges (no
self-loops allowed). Each undirected edge in $E_{c}$ represents a potential
communication channel between nodes $i$ and $j,$ where this channel can be
used to send information in both directions. In this paper, we focus on
directed communications, i.e., the event of node $i$ sending information
towards node $j$ is independent from the even of node $j$ sending
information towards $i$. In this context, it is convenient to interpret an
undirected edge $\left\{ i,j\right\} \in E_{c}$ as the union of two
independent directed edges, $\left\{ \left( i,j\right) ,\left( j,i\right)
\right\} $, where the ordered pair $\left( i,j\right) $ represents a
directed link from node $i$ to node $j$.

In this paper, we study randomized time-switching consensus processes. In
particular, in each discrete time slot $k\geq 1$, we construct a random 
\textit{directed} graph $\mathcal{G}_{k}=\left( V,\mathcal{E}_{k}\right) $,
with $\mathcal{E}_{k}\subseteq E_{c}$, such that the existence of a directed
edge $\left( u,v\right) \in \mathcal{E}_{k}$ is determined randomly and
independently of all other directed edges (including the reciprocal edge $%
\left( v,u\right) $) with a probability $p_{uv}\in \left( 0,1\right) $ for $%
\left( u,v\right) \in E_{c}$, and $p_{uv}=0$ for $\left( u,v\right) \not\in
E_{p}$. In other words, in each time slot, we randomly select a subset $%
\mathcal{E}_{k}$ of directed links chosen from a set of candidate (directed)
links in $E_{c}$. We are specially interested in the case in which the
probability $p_{uv}$ of existence of a directed link $\left( u,v\right) $
depends exclusively on the node that receives information via that link,
i.e., $\Pr \left( \left( u,v\right) \in \mathcal{E}_{k}\right) =p_{v}$,
where $p_{v}\in \left( 0,1\right) $. In this setting, we can model the
ability of a node to `listen' to their neighboring nodes. For example, in
the context of opinion dynamics in social networks \cite{DeGroot1974,
Jackson_Naive},\cite{Asu_misinformation}, the probability $p_{v}$ can
represent the tendency of the individual at node $v$ to take into account
the opinion of her neighbors (which could depend, for example, on how many
acquaintances the individual has).

Let us denote by $A_{c}$ the symmetric adjacency matrix of the graph $G_{c}$%
, where entries $a_{ij}=1\ $if $\left\{ i,j\right\} \in E_{c}$, and $%
a_{ij}=0\ $otherwise. We define the degree of node $i$ as $%
d_{i}=\sum_{j=1}^{n}a_{ji}$, and the associated degree matrix as $%
D_{c}=diag(d_{i})$. We also denote the random (nonsymmetric) adjacency
matrix associated with $\mathcal{G}_{k}$ as $\tilde{A}_{k}=\left[ \tilde{a}%
_{uv}^{\left( k\right) }\right] $, which can be described as%
\begin{equation}
\tilde{a}_{uv}^{\left( k\right) }=\left\{ 
\begin{array}{ll}
a_{uv}, & \text{w.p. }p_{v}, \\ 
0, & \text{w.p. }1-p_{v}.%
\end{array}%
\right.  \label{Random Adjacency}
\end{equation}%
We denote the in-degrees of $\mathcal{G}_{k}$ as $\tilde{d}_{v}^{\left(
k\right) }=\sum_{u}\tilde{a}_{uv}^{\left( k\right) }$, and the in-degree
matrix as $\tilde{D}_{k}=diag(\tilde{d}_{v}^{\left( k\right) })$. From the
definition of $\mathcal{G}_{k}$, the in-degrees are independent Bernoulli
random variables $\tilde{d}_{v}^{\left( k\right) }\sim \Ber\left(
d_{v},p_{v}\right) $, i.e.,%
\begin{equation*}
\Pr \left( \tilde{d}_{v}^{\left( k\right) }=d\right) =\binom{d_{v}}{d}%
p_{v}^{d}\left( 1-p_{v}\right) ^{d_{v}-d}.
\end{equation*}

We describe the consensus dynamics in (\ref{consensus_update}) via a
sequence of stochastic matrices $\{W_{k}\}_{k=1}^{\infty }$ associated to
the sequence of random directed graphs $\left\{ \mathcal{G}_{k}=\left( V,%
\mathcal{E}_{k}\right) \right\} _{k=1}^{\infty }$ as follows:%
\begin{equation}
W_{k}=\left( \tilde{D}_{k}+I\right) ^{-1}\left( \tilde{A}_{k}+I\right) ^{T}%
\text{, for }k\geq 1.  \label{Stochastic Matrix}
\end{equation}%
Notice that adding the identity matrix to the adjacency in (\ref{Stochastic
Matrix}) is equivalent to adding a self-loop to every vertex in $V$. We
include these self-loops to avoid singularities associated with the
existence of isolated nodes in $\mathcal{G}_{k}$ (for which $\tilde{d}_{i}=0$%
, and $\tilde{D}_{k}$ would not invertible). Since $G_{c}$ is connected and $%
p_{v}>0$ for all $v\in V$, the expected communications graph is connected
and the stochastic dynamical system in (\ref{consensus_update}) reaches
consensus on almost all paths, although the asymptotic consensus value $%
x^{\ast }$ is a random variable (not the initial average).

In the following subsections, we first derive closed-form expressions for
the mean of $x^{\ast }$, and an upper-bound for the variance $var\left(
x^{\ast }\right) $ in terms of the following elements:

\begin{enumerate}
\item The set of initial conditions, $\left\{ x_{u}\left( 0\right) \right\}
_{u\in V}$,

\item the set of nodes properties, $\left\{ \left( p_{u},d_{u}\right)
\right\} _{u\in V}$, and

\item the network topology, via the eigenvalues of the expected matrix $%
\mathbb{E}W_{k}$.
\end{enumerate}

As we shall show in Section \ref{Section_Quad_Analysis}, our expression for
the variance has a nice interpretation, since it separates the influence of
each one of the above elements into three multiplicative terms. In the next
subsection we provide the details regarding our analysis of the expectation
of $x^{\ast }$.

\bigskip

\subsection{\label{Start Computations}Mean of Consensus Value}

We use (\ref{mean expression}) to study the mean of the consensus value. We
first derive an expression for $\mathbb{E}W_{k},$ and then study its
dominant left eigenvector $\mathbf{v}_{1}\left( \mathbb{E}W_{k}\right) $.
For notational convenience, we define the random variable $z_{i}\triangleq
1/\left( \tilde{d}_{i}+1\right) $ where $\tilde{d}_{i}\sim \Ber\left(
d_{i},p_{i}\right) $, and denote its first and second moments as $%
M_{i}^{\left( 1\right) }\triangleq \mathbb{E}\left( z_{i}\right) $ and $%
M_{i}^{\left( 2\right) }\triangleq \mathbb{E}\left( z_{i}^{2}\right) $. The
diagonal entries of $\mathbb{E}W_{k}$ are then given by $\mathbb{E}w_{ii}=%
\mathbb{E}\left[ 1/\left( \tilde{d}_{i}+1\right) \right] $, which present
the following explicit expression (see Appendix I for details):%
\begin{equation}
\mathbb{E}w_{ii}=M_{i}^{\left( 1\right) }=\frac{1-q_{i}^{d_{i}+1}}{%
p_{i}\left( d_{i}+1\right) }.  \label{diagonal of EWk}
\end{equation}%
where $q_{i}=1-p_{i}$. Furthermore, the off-diagonal entries of $\mathbb{E}%
W_{k}$ are equal to (see Appendix I for details):%
\begin{align}
\mathbb{E}w_{ij}& =a_{ji}\frac{q_{i}^{d_{i}+1}+p_{i}\left( d_{i}+1\right) -1%
}{p_{i}\left( d_{i}+1\right) d_{i}}  \notag \\
& =a_{ji}\frac{1-M_{i}^{\left( 1\right) }}{d_{i}}.
\label{off-diagonal of EWk}
\end{align}%
Taking (\ref{diagonal of EWk}) and (\ref{off-diagonal of EWk}) into account,
we can write $\mathbb{E}W_{k}$ as follows:%
\begin{equation}
\mathbb{E}W_{k}=\Sigma +\left( I-\Sigma \right) D_{c}^{-1}A_{c}^{T},
\label{expression for EWk}
\end{equation}%
where $\Sigma \triangleq diag\left[ M_{i}^{\left( 1\right) }\right] $. As
expected, it is easy to check that $\mathbb{E}W_{k}$ is a stochastic matrix,
i.e., $\left( \mathbb{E}W_{k}\right) \mathbf{1}_{n}=\Sigma \mathbf{1}%
_{n}+\left( I-\Sigma \right) D_{c}^{-1}A_{c}^{T}\mathbf{1}_{n}=\mathbf{1}%
_{n} $.

Based on (\ref{expression for EWk}) we can write $\mathbb{E}\left( x^{\ast
}\right) $ explicitly in terms of $d_{i}$ and $p_{i}$, as follows:

\begin{theorem}
\label{Asymptotic Mean}Consider the random adjacency matrix $\tilde{A}_{k}$
in (\ref{Random Adjacency}) and the associated (random) stochastic matrix $%
W_{k}$ in (\ref{Stochastic Matrix}). The expectation of the asymptotic
consensus value of (\ref{consensus_update}) is given by 
\begin{equation}
\mathbb{E}\left( x^{\ast }\right) =\sum_{i=1}^{n}\rho w_{i}~x_{i}\left(
0\right) \mathbf{,}  \label{Explicit Expectation}
\end{equation}%
where%
\begin{eqnarray*}
w_{i}\left( p_{i},d_{i}\right) &\triangleq &\frac{p_{i}\left( d_{i}+1\right)
d_{i}}{p_{i}\left( d_{i}+1\right) -1-q_{i}^{d_{i}+1}},\text{ and} \\
\rho \left( p_{i},d_{i}\right) &\triangleq &\left( \sum_{i}w_{i}\left(
p_{i},d_{i}\right) \right) ^{-1}.
\end{eqnarray*}
\end{theorem}

\begin{proof}
Our proof is based on computing $v_{1}\left( \mathbb{E}W_{k}\right) $ and
applying (\ref{mean expression}). Let us define $\mathbf{v}\triangleq
v_{1}\left( \mathbb{E}W_{k}\right) $ and $\mathbf{w}\triangleq \left(
I-\Sigma \right) \mathbf{v}$. From (\ref{expression for EWk}), we have that
the eigenvalue equation corresponding to the dominant left eigenvector of $%
\mathbb{E}W_{k}$ is $\mathbf{v}^{T}\left( \Sigma +\left( I-\Sigma \right)
D_{c}^{-1}A_{c}^{T}\right) =\mathbf{v}^{T}$, which can be rewritten as $%
\mathbf{v}^{T}\left( I-\Sigma \right) D_{c}^{-1}A_{c}^{T}=\mathbf{v}%
^{T}\left( I-\Sigma \right) $. This last equation can be written as $\mathbf{%
w}^{T}D_{c}^{-1}A_{c}^{T}=\mathbf{w}^{T}$. The solution to this equation is
the stationary distribution of the Markov chain with transition matrix $%
D_{c}^{-1}A_{c}^{T}$, which is equal to $\pi =\mathbf{d}/\sum_{i}d_{i}$,
where $\mathbf{d}=\left( d_{1},...,d_{n}\right) ^{T}$. Hence, the solution
to the eigenvalue equation is%
\begin{equation*}
v_{1}\left( \mathbb{E}W_{k}\right) =\sigma \left( I-\Sigma \right) ^{-1}%
\mathbf{d,}
\end{equation*}%
where we include the normalizing parameter 
\begin{equation*}
\sigma =\left( \sum_{i}\frac{d_{i}}{1-M_{i}^{\left( 1\right) }}\right) ^{-1}
\end{equation*}%
such that $\left\Vert v_{1}\left( \mathbb{E}W_{k}\right) \right\Vert _{1}=1$%
. Hence, from (\ref{mean expression}) we have%
\begin{equation*}
\mathbb{E}\left( x^{\ast }\right) =\sum_{i=1}^{n}\sigma \frac{d_{i}}{%
1-M_{i}^{\left( 1\right) }}x_{i}\left( 0\right) .
\end{equation*}%
Substituting the expression for $M_{i}^{\left( 1\right) }$ in (\ref{diagonal
of EWk}), we reach (\ref{Explicit Expectation}) via simple algebraic
simplifications.
\end{proof}

\bigskip

In general, the asymptotic mean $\mathbb{E}\left( x^{\ast }\right) $ does
not coincide with the initial average $\bar{x}_{0}=\frac{1}{n}%
\sum_{i}x_{i}\left( 0\right) $. There is a simple technique, based on
Theorem \ref{Asymptotic Mean}, that allows us to make the expected consensus
value to be equal to the initial average. This technique consists of using $%
y_{i}\left( 0\right) =\left( \rho w_{i}\right) ^{-1}x_{i}\left( 0\right) $
as initial conditions in (\ref{consensus_update}). Hence, one can easily
check that the asymptotic consensus value $\mathbb{E}\left( y^{\ast }\right) 
$ equals the initial average $\bar{x}_{0}$.

\bigskip

\section{Variance of the Asymptotic Consensus Value}

\label{Section_Quad_Analysis}

In this section, we derive an expression that explicitly relates the
variance $var\left( x^{\ast }\right) $ with the three elements that
influences it, namely, the set of initial conditions $\left\{ x_{u}\left(
0\right) \right\} _{u\in V}$, the nodes properties $\left\{ \left(
p_{u},d_{u}\right) \right\} _{u\in V}$, and the network structure (via the
eigenvalues of the expected matrix $\mathbb{E}W_{k}$). For simplicity in
notation, we denote $\mathbb{E}\left( W_{k}\otimes W_{k}\right) $ and $%
\left( \mathbb{E}W_{k}\otimes \mathbb{E}W_{k}\right) $ by $R$ and $Q$,
respectively. Our analysis starts in expression (\ref{variance_expression}),
which can be rewritten as%
\begin{equation*}
var\left( x^{\ast }\right) =\left[ \mathbf{x}(0)\otimes \mathbf{x}(0)\right]
^{T}\left[ v_{1}\left( R\right) -v_{1}\left( Q\right) \right] .
\end{equation*}
Hence, we can upper-bound the variance of the asymptotic consensus value as
follows:%
\begin{equation}
var\left( x^{\ast }\right) \leq \left\Vert \mathbf{x}(0)\otimes \mathbf{x}%
(0)\right\Vert _{1}\left\Vert v_{1}\left( R\right) -v_{1}\left( Q\right)
\right\Vert _{\infty }.  \label{variance bound}
\end{equation}%
From the rules of Kronecker multiplication, we can write the first factor in
terms of the initial conditions as%
\begin{equation}
\left\Vert \mathbf{x}(0)\otimes \mathbf{x}(0)\right\Vert _{1}=\sum_{1\leq
i,j\leq n}\left\vert x_{i}\left( 0\right) x_{j}\left( 0\right) \right\vert .
\label{Initial Conditions Term}
\end{equation}%
In the following, we derive an upper bound for the second factor $\left\Vert
v_{1}\left( R\right) -v_{1}\left( Q\right) \right\Vert _{\infty }$ in terms
of the nodes properties and the network structure. Our approach to bound $%
\left\Vert v_{1}\left( R\right) -v_{1}\left( Q\right) \right\Vert _{\infty }$
is based on the observation that both $R$ and $Q$ are $n^{2}\times n^{2}$
stochastic matrices, and the dominant left eigenvectors $\mathbf{v}%
_{1}\left( R\right) $ and $\mathbf{v}_{1}\left( Q\right) $ are stationary
distributions of the Markov chains with transition matrices $R$ and $Q$. We
denote these distributions by $\mathbf{v}_{1}\left( R\right) \triangleq 
\tilde{\pi}$ and $\mathbf{v}_{1}\left( Q\right) \triangleq \pi $,
respectively. In this setting, we can apply the following lemma from \cite%
{FM86} which studies the sensitivity of the stationary distribution of
Markov chains:

\begin{lemma}
\label{Cho-Meyer}Consider two Markov chains with transition matrices $Q$ and 
$R$, and stationary distributions $\pi $ and $\tilde{\pi}$, respectively. We
define $G=I-Q,$ and denote its pseudoinverse by $G^{\dag }=[g_{ij}^{\dag }]$%
. Hence,%
\begin{equation}
\left\Vert \tilde{\pi}-\pi \right\Vert _{\infty }\leq \kappa _{s}\left\Vert
R-Q\right\Vert _{\infty },  \label{Cho-Meyer Bound}
\end{equation}%
where $\kappa _{s}=\max_{i,j}\left\vert g_{ij}^{\dag }\right\vert $ is
called the condition number of the chain described by the transition matrix $%
Q$.
\end{lemma}

In the next subsections, apply the above lemma to bound the factor $%
\left\Vert v_{1}\left( R\right) -v_{1}\left( Q\right) \right\Vert _{\infty }$%
. In the first subsection we compute an explicit expression for the norm of
the perturbation $\left\Vert R-Q\right\Vert _{\infty }$ as a function of the
properties of the nodes. In the second subsection, we study the coefficient $%
\kappa _{s}$ in terms of the eigenvalues of $\mathbb{E}W_{k}$.

\bigskip

\subsection{Infinity Norm of the Perturbation $\left\Vert R-Q\right\Vert
_{\infty }$}

Our approach is based on studying the entries of the $n^{2}\times n^{2}$
matrix $R=\mathbb{E}\left[ W_{k}\otimes W_{k}\right] $, and compare them
with the entries of $Q=\mathbb{E}W_{k}\otimes \mathbb{E}W_{k}$. The entries
of $Q$ and $R$ are of the form $\mathbb{E}(w_{ij})\mathbb{E}(w_{rs})$ and $%
\mathbb{E}(w_{ij}w_{rs})$, respectively, with $i,j,r,$ and $s$ ranging from $%
1$ to $n$. These entries can be classified into seven different cases
depending on the relations between the indices. In the Appendices II\ and
III, we present explicit expressions for each one of these cases. A key
observation in our approach is to notice the pattern that the above entries
induces in the matrices $R$ and $Q$. For sake of clarity, we illustrate this
pattern for $n=3$ in Fig. 1, where the numbers in parenthesis correspond to
each one of the seven cases identified in the Appendices.

\begin{figure}[tbp]
\centering
\includegraphics[width=0.80\linewidth]{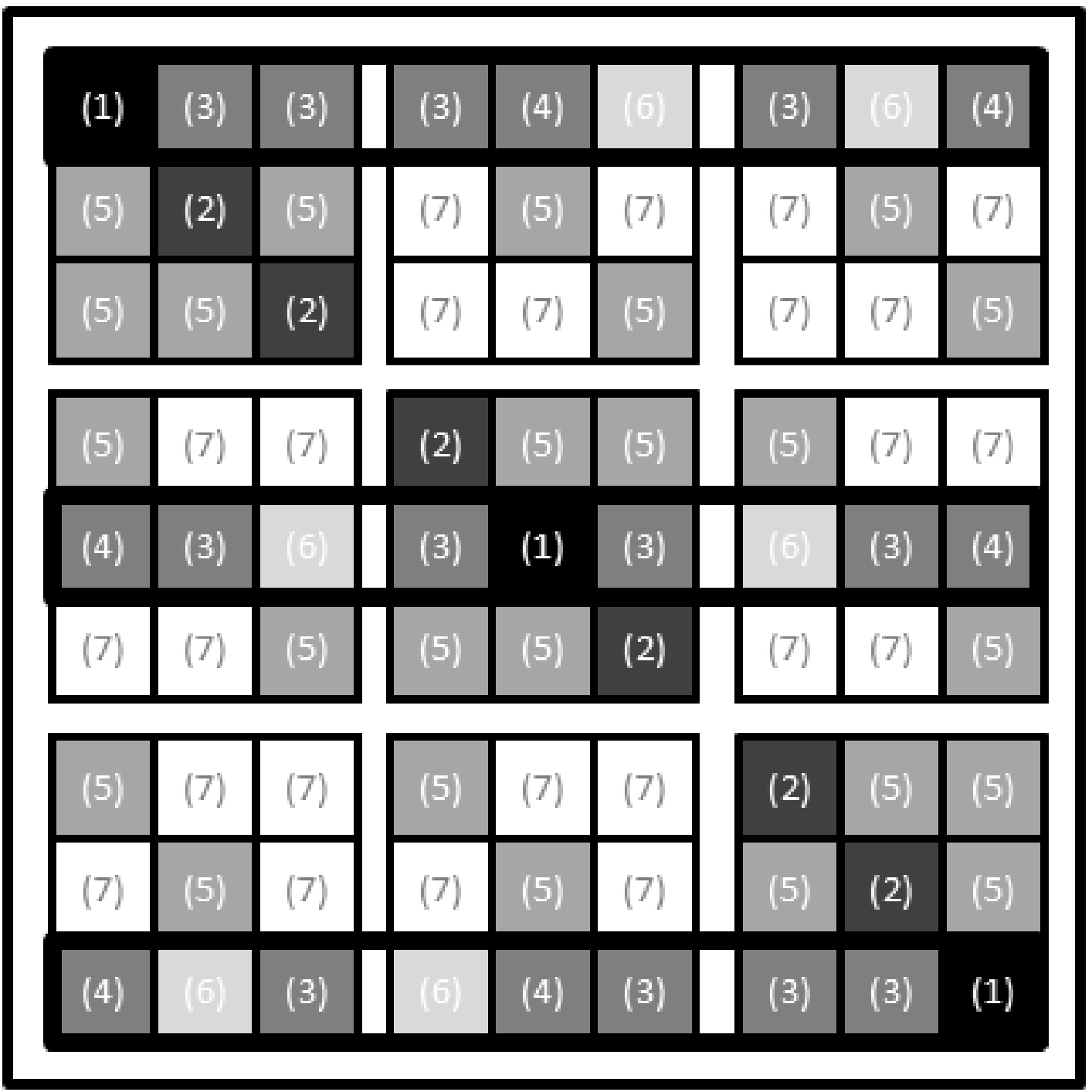} 
\caption{The pattern of $\mathbb{E}[W_k\otimes W_k]$ for $n=3$. The numbers
in parentheses represent the value of each entry defined in the Appendices.}
\end{figure}

Since the entries of $R$ follow the same pattern as the entries of $Q$
(although these entries are different), the perturbation matrix $\Delta =R-Q$
also follows the same pattern as $R$ and $Q$. Hence, comparing the entries
of $Q$ and $R$ we can easily deduce the following seven cases for the
entries of $\Delta $:%
\begin{align*}
\Delta _{2}& =\Delta _{5}=\Delta _{7}=0, \\
\Delta _{1}& =M_{2}\left( z_{i}\right) -M_{1}^{2}\left( z_{i}\right) , \\
\Delta _{3}& =\frac{a_{ji}}{d_{i}}\left( -M_{2}\left( z_{i}\right)
+M_{1}^{2}\left( z_{i}\right) \right) , \\
\Delta _{4}& =a_{ji}\left( \frac{M_{1}\left( z_{i}\right) -M_{2}\left(
z_{i}\right) }{d_{i}}-\frac{\left( 1-M_{1}\left( z_{i}\right) \right) ^{2}}{%
d_{i}^{2}}\right) , \\
\Delta _{6}& =\left\{ 
\begin{array}{ll}
\left( \frac{\left( 1-M_{i}\right) ^{2}}{d_{i}^{2}}-\frac{1+2V_{i}-3M_{i}}{%
d_{i}\left( d_{i}-1\right) }\right) , & \text{for }d_{i}>1, \\ 
-a_{ji}a_{si}\left( 1-M_{1}\left( z_{i}\right) \right) ^{2}, & \text{for }%
d_{i}=1.%
\end{array}%
\right.
\end{align*}%
where $M_{i}^{\left( 2\right) }\triangleq \mathbb{E}\left[ 1/(\tilde{d}%
_{i}+1)^{2}\right] $ can be written as a hypergeometric function that
depends on $p_{i}$ and $d_{i}$ (see Appendix III). From the above entries,
we can compute via simple, but tedious, algebraic manipulations the
following expression for the infinity norm of the perturbation:%
\begin{equation}
\left\Vert R-Q\right\Vert _{\infty }=\left\Vert \Delta \right\Vert _{\infty
}=\max_{1\leq i\leq n}\left\{ S_{i}\right\} ,  \label{Inf Norm Bound}
\end{equation}%
where%
\begin{equation*}
S_{i}=2(1-M_{i}^{\left( 1\right) })\left[ M_{i}^{\left( 1\right) }+\frac{1}{%
d_{i}}\left( M_{i}^{\left( 1\right) }-1\right) \right] ,
\end{equation*}%
Note that $S_{i}$ is a function that depends exclusively on the sequence of
nodes properties $\left\{ \left( p_{u},d_{u}\right) \right\} _{u\in V}$;
hence, $\left\Vert R-Q\right\Vert _{\infty }$ depends exclusively on the set
of nodes degrees and probabilities, but it is independent on how these
nodes interconnect. In the following subsection, we show how the pattern of
interconnection among nodes influences the upper-bound of the variance in (%
\ref{variance bound}) via the condition number $\kappa _{s}$ in (\ref%
{Cho-Meyer Bound}).

\bigskip

\subsection{Perturbation-Based Bound for the Variance}

In this subsection we derive an explicit relationship between the condition
number $\kappa _{s}$ and the network structure via the spectral properties
of $\mathbb{E}W_{k}=\Sigma +\left( I-\Sigma \right) D_{c}^{-1}A_{c}^{T}$.
This result will then be used to bound the variance of $x^{\ast }$ in (\ref%
{variance bound}). We base our analysis on a the following bound, derived by
Meyer in \cite{Meyer94}, relating $\kappa _{s}$ with the eigenvalues of $Q$,
denoted as $\left\{ \mu _{i}\right\} _{i=1}^{n^{2}}$:%
\begin{equation}
\max_{i,j}\left\vert g_{ij}^{\dag }\right\vert <\frac{2\left( n^{2}-1\right) 
}{\prod_{k=2}^{n^{2}}\left( 1-\mu _{k}\right) }.
\label{General Spectral Bound}
\end{equation}

Before we present our main result, we need some notation. Denote by $\left\{
\lambda _{i}\right\} _{i=1}^{n}$ and $\left\{ \mu _{j}\right\}
_{j=1}^{n^{2}} $ the set of eigenvalues of $\mathbb{E}W_{k}$ and $\mathbb{E}%
W_{k}\otimes \mathbb{E}W_{k}$, respectively. The ordering of the eigenvalue
sequences is determined by their distance to $1$, i.e., $\left\vert
1-\lambda _{i}\right\vert \leq \left\vert 1-\lambda _{j}\right\vert $ for $%
i\leq j$. Hence, our result can be stated as follows:

\begin{theorem}
\label{Bound Asymptotic Variance}The variance of the asymptotic consensus
value of (\ref{consensus_update}) can be upper-bounded by%
\begin{equation}
var\left( x^{\ast }\right) \leq \underset{(A)}{\underbrace{\left\Vert 
\mathbf{x}(0)\otimes \mathbf{x}(0)\right\Vert _{1}}}\underset{(B)}{%
\underbrace{\left( \max_{1\leq i\leq n}\left\{ S_{i}\right\} \right) }}%
\underset{(C)}{\underbrace{\left( \frac{2\left( n^{2}-1\right) }{\prod
\left( 1-\lambda _{i}\lambda _{j}\right) }\right) }},
\label{Separation of terms}
\end{equation}%
where $\left\{ \lambda _{i}\right\} _{i=1}^{n}$ are the eigenvalues of $%
\mathbb{E}W_{k},$ and the product $\prod \left( 1-\lambda _{i}\lambda
_{j}\right) =\prod_{\left( i,j\right) \text{ s.t. }\left( i,j\right) \neq
\left( 1,1\right) }\left( 1-\lambda _{i}\lambda _{j}\right) .$
\end{theorem}

\bigskip

\begin{proof}
We start our proof from (\ref{variance bound})%
\begin{eqnarray*}
var\left( x^{\ast }\right)  &\leq &\left\Vert v_{1}\left( R\right)
-v_{1}\left( Q\right) \right\Vert _{\infty }\left\Vert \mathbf{x}(0)\otimes 
\mathbf{x}(0)\right\Vert _{1} \\
&&\overset{(a)}{\leq }\kappa _{s}\left\Vert R-Q\right\Vert _{\infty
}\left\Vert \mathbf{x}(0)\otimes \mathbf{x}(0)\right\Vert _{1} \\
&&\overset{(b)}{=}\kappa _{s}\left( \max_{1\leq i\leq n}\left\{
S_{i}\right\} \right) \left\Vert \mathbf{x}(0)\otimes \mathbf{x}%
(0)\right\Vert _{1} \\
&&\overset{(c)}{<}\frac{2\left( n^{2}-1\right) }{\prod \left( 1-\mu
_{k}\right) }\max_{1\leq i\leq n}\left\{ S_{i}\right\} \left\Vert \mathbf{x}%
(0)\otimes \mathbf{x}(0)\right\Vert _{1},
\end{eqnarray*}%
where we have used Lemma \ref{Cho-Meyer} in inequality (a), the expression
for $\left\Vert R-Q\right\Vert _{\infty }$ in (\ref{Inf Norm Bound}) in
equality (b), and the upper bound in (\ref{General Spectral Bound}) in step
(c). We obtain the statement of the theorem by applying the following
standard property of the Kronecker product: $\left\{ \mu _{k}\right\}
_{k=1}^{n^{2}}=\left\{ \lambda _{i}\lambda _{j}\right\} _{1\leq i,j<n}$.
\end{proof}

\bigskip

\begin{remark}
The bound in (\ref{Separation of terms}) separates the variance into three
multiplicative terms representing each one of the following elements (for
convenience, we have underlined each one of these terms in (\ref{Separation
of terms})):

\begin{description}
\item[(A)] This first term exclusively depends on the initial condition as
indicated by (\ref{Initial Conditions Term}).

\item[(B)] The second term depends solely on the nodes properties $d_{i}$
and $p_{i}$.

\item[(C)] The last term represents the influence from the overall graph
structure via the spectral properties of $\mathbb{E}W_{k}$.
\end{description}
\end{remark}

\bigskip

\begin{remark}
\label{Spectral Implications}It is specially interesting to study the
implications of Term (C) in the asymptotic variance. For example, given the
sequences of degrees and probabilities, $\left\{ d_{i}\right\} _{i=1}^{n}$
and $\left\{ p_{i}\right\} _{i=1}^{n}$, the influence of the network
structure on the variance is given via Term (C). Since the eigenvalues of $%
\mathbb{E}W_{k}$ are key in this term, it is interesting to briefly describe
the homogeneous Markov chain with transition matrix\ $P\triangleq \mathbb{E}%
W_{k}=\Sigma +\left( I-\Sigma \right) D_{c}^{-1}A_{c}^{T}$. This Markov
chain presents a self-loop in each one of its $n$ states with transition
probability $P_{ii}=M_{i}^{\left( 1\right) }$, as well as a link from $i$ to 
$j$ with transition probability $P_{ij}=\left( 1-M_{i}^{\left( 1\right)
}\right) d_{i}^{-1}$ for $\left( i,j\right) \in E_{c}$. From an analysis
proposed by Meyer in \cite{Meyer94}, we conclude that Term (C) is primarily
governed by how close the subdominant eigenvalues of $\mathbb{E}W_{k}$ are
from $1$. In particular, the further the subdominant eigenvalues of $\mathbb{%
E}W_{k}$ are from 1, the lower the upper bound in (\ref{Separation of terms}%
). In the next section, we illustrate the relationship between spectral
properties of $\mathbb{E}W_{k}$\ and the variance of the asymptotic
consensus value with several numerical simulations.
\end{remark}

\bigskip

\section{Numerical Simulations}

In this section we present several numerical simulations illustrating our
results. In the first subsection, we numerically verify the result in
Theorem \ref{Asymptotic Mean}. In the second subsection we present some
examples to illustrate the implications of Theorem \ref{Bound Asymptotic
Variance} in the variance of the asymptotic consensus value.

\subsection{Expectation of the Asymptotic Consensus}

In our first simulation we take a graph $G_{c}$ composed by 3 stars
connected in a chain (see Fig. 2). This graph is intended to represent, in a
social context, three leaders with a set of followers. In Fig. 2, we
represent the leaders using large circles marked as 1,2, and 3, and the
followers using smaller circles. We assume that each follower only listen to
one leader (the center of a star) and nobody else. In this particular
example, the first, second and third leaders have $4$, $8$, and $16$
followers, respectively. In each time step, a directed random graph $%
\mathcal{G}_{k}$ is built by choosing a set of directed communication links
from $G_{c}$. We have fixed the probability of existence of a directed link
coming into followers to be equal to $p_{lf}=1/2$, for all such links (see
Fig. 2). Also, the probability of existence of a directed link coming into a
leader is inversely proportional to the degree of the leader. More
specifically, we have chosen $p_{uv}=1/d_{v}$, for $v=1,2,$ and $3$, where $%
p_{uv}$ represents the existence of a directed communication link from the $u
$-th node to the $v$-th leader.

\begin{figure}[tbp]
\centering
\includegraphics[width=0.95\linewidth]{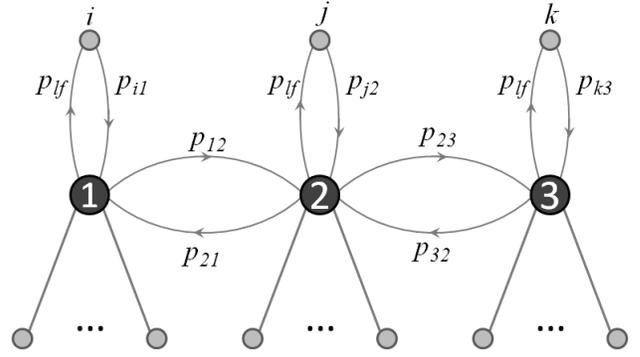} 
\caption{The above figure represents a chain of leader nodes, denoted as
1,2, and 3, with a set of followers. Connections between leaders are
represented as pairs of directed links with existence probabilities $p_{uv}$
and $p_{vu}$, with $u,v\in \left\{ 1,2,3\right\} $. We have also included in
this figure the connection of each leader and one of its followers (drawn as
a smaller circle on top of each leader). In this case, the directed link
incoming the follower (leader) exists with probability $p_{lf}$ ($p_{fl}$).
The rest of followers are represented under each leader, where we represent
each pair of directed links as a single undirected edge for clarity.}
\end{figure}

In this example, we compute the asymptotic consensus value for 100
realizations of a random consensus algorithm with initial conditions $%
x_{i}\left( 0\right) =i/N$. We represent the histogram of these realizations
in Fig. 3, where the empirical average of the 100 realizations equals
0.5077. The theoretical expectation for the asymptotic consensus value
applying Theorem \ref{Asymptotic Mean} equals $\mathbb{E}\left( x^{\ast
}\right) =\sum_{i}\rho w_{i}\left( i/N\right) =$0.4595, which is in great
accordance with the empirical value. For future reference, we have also
computed the empirical standard deviation of the 100 realizations to be
0.0298, and the three eigenvalues of $\mathbb{E}W_{k}$ closest to 1 are $%
\left\{ 0.9823,0.9449,0.75\right\} $.

\begin{figure}[tbp]
\centering
\includegraphics[width=0.85\linewidth]{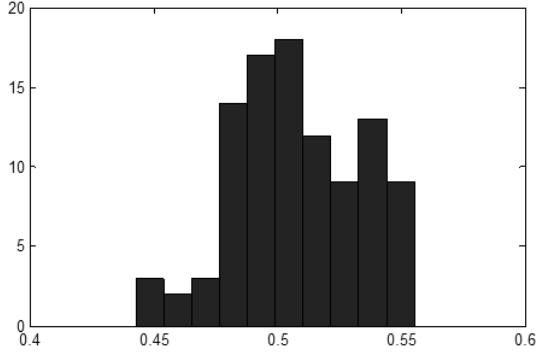} 
\caption{Empirical histogram for 100 realizations of the random directed
consensus over the graph in Fig. 2. The empirical average and the empirical
standard deviation of the realizations are 0.4595 and 0.0298, respectively.}
\end{figure}

\subsection{Expectation of the Asymptotic Variance}

In this subsection, we numerically illustrate some of the implications of
Theorem \ref{Bound Asymptotic Variance} in the variance of the asymptotic
consensus value. Although the upper bound stated in this Theorem is not
tight, there are some qualitative implications that are consistent
throughout our numerical experiments. For example, as we mentioned in Remark %
\ref{Spectral Implications}, given a set of initial conditions and nodes
properties, the upper bound in (\ref{Separation of terms}) is primarily
governed by how close the subdominant eigenvalues of $\mathbb{E}W_{k}$ are
from $1$. In particular, the further the subdominant eigenvalues of $\mathbb{%
E}W_{k}$ are from 1, the lower the upper bound. We illustrate in the
following numerical experiments that the lower the upper bound, the lower
one should expect the variance to be.

In the first experiment, we slightly modify the network described in the
previous subsection and study the influence of this modification on the
eigenvalues of $\mathbb{E}W_{k}$ and the variance of $x^{\ast }$. Our first
modification is a change in the probabilities of existence of a directed
link from the $u$-th node to the $v$-th leader without changing the network
topology. In particular, we choose the new probabilities to be $%
p_{uv}=3/d_{v}$, for $v=1,2,$ and $3$. In the modified network, the
eigenvalues of $\mathbb{E}W_{k}$ closest to 1 are $\left\{
0.9789,0.9372,0.75\right\} $. Hence, since the effect of our modification on
the eigenvalues is to move them away from 1, we should expect the variance
of $x^{\ast }$ to be reduced according to Remark \ref{Spectral Implications}%
. In fact, running 100 random consensus algorithms with the same initial
conditions, $x_{i}\left( 0\right) =i/N,$ using our new probabilities, we
obtain a standard deviation 0.0286 (which is less than the 0.0298 obtained
before).

In the second experiment, we illustrate how the larger the gap between the
eigenvalues of $\mathbb{E}W_{k}$ and $1$, the smaller the variance of $%
x^{\ast }$. In this case, apart from keeping the new set of probabilities $%
p_{uv}=3/d_{v},$ we convert the 3-chain of leaders into a 3-ring of leaders,
as depicted in Fig. 4. In this case, the three eigenvalues of $EW_{k}$
closest to $1$ are $\left\{ 0.9577,0.9212,0.75\right\} $, which are even
further away from 1 than in the second example. Hence, as expected, the
standard deviation of 100 random consensus algorithms is even smaller than
in the second example, in particular 0.0274. In conclusion, our simulations
are consistent with Theorem \ref{Bound Asymptotic Variance} and with the
qualitative behavior described in Remark \ref{Spectral Implications}.

\begin{figure}[tbp]
\centering
\includegraphics[width=0.95\linewidth]{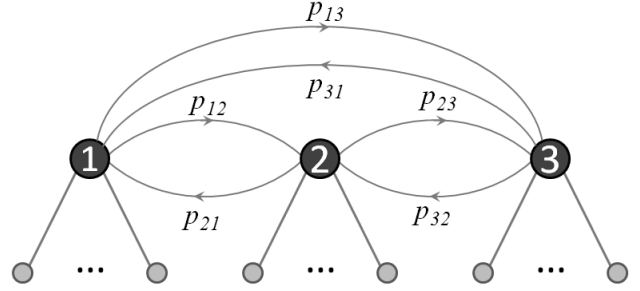} 
\caption{Ring of 3 leaders and corresponding followers.}
\end{figure}

\section{Conclusions and Future Work}

We have studied the asymptotic properties of the consensus value in
distributed consensus algorithms over switching, directed random graphs.
While different aspects of consensus algorithms over random switching
networks, such as conditions for convergence and the speed of convergence,
have been widely studied, a characterization of the distribution of the
asymptotic consensus for general \textit{asymmetric} random consensus
algorithms remains an open problem.

In this paper, we have derived closed-form expressions for the expectation
of the asymptotic consensus value as a function of the set of initial
conditions, $\left\{ x_{u}\left( 0\right) \right\} _{u\in V}$, and the set
of nodes properties, $\left\{ \left( p_{u},d_{u}\right) \right\} _{u\in V}$,
as stated in Theorem \ref{Asymptotic Mean}. We have also studied the
variance of the asymptotic consensus value in terms of several elements that
influence it, namely, (\emph{i}) the initial conditions, (\emph{ii}) nodes
properties, and (\emph{iii}) the network topology. In Theorem \ref{Bound
Asymptotic Variance}, we have derived an upper bound for the variance of the
asymptotic consensus value that explicitly describes the influence of each
one of these elements. We also provide an interpretation of the influce of
the network topology on the variance in terms of the eigenvalues of the
expected matrix $\mathbb{E}W_{k}$. From our analysis we conclude that, in
most cases, the variance of $x^{\ast }$ is primarily governed by how close
the subdominant eigenvalues of $\mathbb{E}W_{k}$ are from $1$. We have
checked the validity of our predictions with several numerical
simulations.\bigskip 

\appendix[Computing the Entries of $\mathbb{E}W_{k}$]

We start by computing the entries of $\mathbb{E}W_{k}$. The diagonal entries
of $\mathbb{E}W_{k}$ are given by: 
\begin{eqnarray*}
\mathbb{E}w_{ii} &=&\mathbb{E}\left[ \frac{1}{1+\tilde{d}_{i}}\right]
=\sum_{k=0}^{d_{i}}\frac{1}{k+1}\mathbb{P}(\tilde{d}_{i}=k) \\
&=&\sum_{k=0}^{d_{i}}\frac{1}{k+1}\binom{d_{i}}{k}%
p_{i}^{k}(1-p_{i})^{d_{i}-k} \\
&=&\frac{1-q_{i}^{d_{i}+1}}{p_{i}\left( d_{i}+1\right) }\triangleq
M_{i}^{\left( 1\right) }
\end{eqnarray*}%
Similarly, the non-diagonal entries of $\mathbb{E}W_{k}$ result in:%
\begin{eqnarray*}
\mathbb{E}w_{ij} &=&\mathbb{E}\left[ \frac{\tilde{a}_{ji}}{1+\tilde{d}_{i}}%
\right] \\
&=&a_{ji}\frac{q_{i}^{d_{i}+1}+p_{i}\left( d_{i}+1\right) -1}{p_{i}\left(
d_{i}+1\right) d_{i}} \\
&=&a_{ji}\frac{1-M_{i}^{\left( 1\right) }}{d_{i}}.
\end{eqnarray*}

\bigskip

\appendix[Entries of $\mathbb{E}W_{k} \otimes \mathbb{E}W_{k}$]

We now compute the possible entries in $Q=\mathbb{E}W_{k}\otimes \mathbb{E}%
W_{k}$. The entries of the Kronecker matrix $\mathbb{E}W_{k}\otimes \mathbb{E%
}W_{k}$, with $\mathbb{E}W_{k}=\Sigma +\left( I-\Sigma \right) D^{-1}A^{T}$,
present entries that can be classified into seven different cases depending
on the relations between the indices. These are the cases, where we assume
that all four indices $i,j,r,$ and $s$ are distinct:%
\begin{align*}
Q_{1}& \triangleq \mathbb{E}\left( w_{ii}\right) \mathbb{E}\left(
w_{ii}\right) =\left( M_{i}^{\left( 1\right) }\right) ^{2}, \\
Q_{2}& \triangleq \mathbb{E}\left( w_{ii}\right) \mathbb{E}\left(
w_{jj}\right) =M_{i}^{\left( 1\right) }M_{j}^{\left( 1\right) }, \\
Q_{3}& \triangleq \mathbb{E}\left( w_{ii}\right) \mathbb{E}\left(
w_{ij}\right) =\frac{a_{ji}}{d_{i}}M_{i}^{\left( 1\right) }\left(
1-M_{i}^{\left( 1\right) }\right) , \\
Q_{4}& \triangleq \mathbb{E}\left( w_{ij}\right) \mathbb{E}\left(
w_{ij}\right) =\frac{a_{ji}}{d_{i}^{2}}\left( 1-M_{i}^{\left( 1\right)
}\right) ^{2}, \\
Q_{5}& \triangleq \mathbb{E}\left( w_{ii}\right) \mathbb{E}\left(
w_{rs}\right) =\frac{a_{sr}}{d_{r}}~M_{i}^{\left( 1\right) }\left( \text{~}%
1-M_{r}^{\left( 1\right) }\right) , \\
Q_{6}& \triangleq \mathbb{E}\left( w_{ij}\right) \mathbb{E}\left(
w_{is}\right) =\frac{a_{ji}a_{si}}{d_{i}^{2}}\left( 1-M_{i}^{\left( 1\right)
}\right) ^{2}, \\
Q_{7}& \triangleq \mathbb{E}\left( w_{ij}\right) \mathbb{E}\left(
w_{rs}\right) =\frac{a_{ji}a_{sr}~}{d_{i}d_{r}}\left( 1-M_{i}^{\left(
1\right) }\right) \left( ~1-M_{r}^{\left( 1\right) }\right) .
\end{align*}

\bigskip

\appendix[Entries of $\mathbb{E}W_{k} \otimes W_{k}$]

We now turn to the computation of the elements of $\mathbb{E}[W_{k}\otimes
W_{k}]$, which are of the form $\mathbb{E}(w_{ij}w_{rs})$. Again, we
classify the entries into seven different cases depending on the relations
between the subindices:

\begin{eqnarray*}
R_{1} &\triangleq &\mathbb{E}w_{ii}^{2}=\mathbb{E}\left[ \frac{1}{(\tilde{d}%
_{i}+1)^{2}}\right] \\
&=&\sum_{k=0}^{d_{i}}\frac{1}{(k+1)^{2}}\binom{d_{i}}{k}%
p_{i}^{k}q_{i}^{d_{i}-k} \\
&=&q^{d_{i}}H(p_{i},d_{i})\triangleq M_{i}^{\left( 2\right) }
\end{eqnarray*}%
Similarly, we also have:%
\begin{eqnarray*}
R_{2} &\triangleq &\mathbb{E}\left( w_{ii}w_{jj}\right) =M_{i}^{\left(
1\right) }M_{j}^{\left( 1\right) }, \\
R_{3} &\triangleq &\mathbb{E}\left( w_{ii}w_{ij}\right) =\frac{a_{ji}}{d_{i}}%
\left( M_{i}^{\left( 1\right) }-M_{i}^{\left( 2\right) }\right) , \\
R_{4} &\triangleq &\mathbb{E}\left( w_{ij}w_{ij}\right) =\frac{a_{ji}}{d_{i}}%
\left( M_{i}^{\left( 1\right) }-M_{i}^{\left( 2\right) }\right) , \\
R_{5} &\triangleq &\mathbb{E}\left( w_{ii}w_{rs}\right) =\frac{a_{sr}}{d_{r}}%
~M_{i}^{\left( 1\right) }\left( \text{~}1-M_{r}^{\left( 1\right) }\right) ,
\\
R_{6} &\triangleq &\mathbb{E}\left( w_{ij}w_{is}\right) \\
&=&\left\{ 
\begin{tabular}{ll}
$\frac{a_{ji}a_{si}}{d_{i}\left( d_{i}-1\right) }\left( 1+2M_{i}^{\left(
2\right) }-3M_{i}^{\left( 1\right) }\right) ,$ & for $d_{i}>1,$ \\ 
$0,$ & for $d_{i}=1,$%
\end{tabular}%
\right. \\
R_{7} &\triangleq &\mathbb{E}(w_{ij}w_{ji})=\mathbb{E}(w_{ij}w_{js})=\mathbb{%
E}(w_{ij}w_{ri}) \\
&=&\mathbb{E}(w_{ij}w_{rj})=\mathbb{E}(w_{ij}w_{rs}) \\
&=&\frac{a_{ji}a_{sr}~}{d_{i}d_{r}}\left( 1-M_{i}^{\left( 1\right) }\right)
\left( ~1-M_{r}^{\left( 1\right) }\right) .
\end{eqnarray*}




\begin{thebibliography}{99}
\bibitem{TsitsiklisThesis} J. N. Tsitsiklis, \emph{Problems in decentralized
decision making and computation}, Ph.D. dissertation, Massachusetts
Institute of Technology, Cambridge, MA, 1984.

\bibitem{Jad2003} A. Jadbabaie, J. Lin, and A. S. Morse, \textquotedblleft
Coordination of groups of mobile autonomous agents using nearest neighbor
rules," \emph{IEEE Transactions on Automatic Control}, vol. 48, no. 6, pp.
988-1001, 2003.

\bibitem{Olfati2007} R. Olfati-Saber, J. A. Fax, and R. M. Murray,
\textquotedblleft Consensus and cooperation in networked multi-agent
systems," \emph{Proceedings of the IEEE}, vol. 95, no. 1, 2007.

\bibitem{Bullo2005} J. Cortes, S. Martinez, and F. Bullo, \textquotedblleft
Analysis and design tools for distributed motion coordination," in \emph{%
Proceedings of the American Control Conference}, Portland, OR, pp.
1680-1685, 2005.

\bibitem{DeGroot1974} M. H. DeGroot, \textquotedblleft Reaching a
Consensus," \emph{Journal of American Statistical Association}, vol. 69, no.
345, pp. 118-121, 1974.

\bibitem{Jackson_Naive} B. Golub and M. O. Jackson, \textquotedblleft Naive
learning in social networks: Convergence, influence, and the wisdom of
crowds," 2007, unpublished Manuscript.

\bibitem{Hatano2004} Y. Hatano and M. Mesbahi, \textquotedblleft Agreement
over random networks," \emph{IEEE Transactions on Automatic Control}, vol.
50, no. 11, pp. 1867-1872, 2005.

\bibitem{WuRandomConsensus} C. W. Wu, \textquotedblleft Synchronization and
convergence of linear dynamics in random directed networks," \emph{IEEE
Transactions on Automatic Control}, vol. 51, no. 7, pp. 1207-1210, 2006.

\bibitem{Porfiri_TAC} M. Porfiri and D. J. Stilwell, \textquotedblleft
Consensus seeking over random weighted directed graphs," \emph{IEEE
Transactions on Automatic Control}, vol. 52, no. 9, pp. 1767-1773, 2007.

\bibitem{Tahbaz_Jad_TAC} A. Tahbaz-Salehi and A. Jadbabaie,
\textquotedblleft A necessary and sufficient condition for consensus over
random networks," \emph{IEEE Transactions on Automatic Control}, vol. 53,
no. 3, pp. 791-795, 2008.

\bibitem{Picci_Taylor_CDC} G. Picci and T. Taylor, \textquotedblleft Almost
sure convergence of random gossip algorithms," in \emph{Proceedings of the
46th IEEE Conference on Decision and Control}, New Orleans, LA, pp. 282-287,
2007.

\bibitem{Asu_misinformation} D. Acemoglu, A. Ozdaglar, and A.
ParandehGheibi, \textquotedblleft Spread of (mis)information in social
networks," \emph{Massachusetts Institute of Technology, LIDS Report 2812},
May 2009, submitted for Publication.

\bibitem{Zampieri} F. Fagnani and S. Zampieri, \textquotedblleft Randomized
consensus algorithms over large scale networks," \emph{IEEE Journal on
Selected Areas in Communications}, vol. 26, no. 4, pp. 634-649, 2008.

\bibitem{Boyd2005b} S. Boyd, A. Gosh, B. Prabhakar, and D. Shah,
\textquotedblleft Randomized gossip algorithms," \emph{Special issue of IEEE
Transactions on Information Theory and IEEE/ACM Transactions on Networking},
vol. 52, no. 6, pp. 2508-2530, 2006.

\bibitem{Tahbaz_Jad_TAC_2} A. Tahbaz-Salehi and A. Jadbabaie,
\textquotedblleft Consensus over ergodic stationary graph processes," \emph{%
IEEE Transactions on Automatic Control}, vol. 55, no. 1, pp. 225-230, 2010.

\bibitem{PTJ10} V.M. Preciado, A. Tahbaz-Salehi, and A. Jadbabaie,
\textquotedblleft Variance analysis of randomized consensus in switching
directed networks,"in \emph{Proceedings of the IEEE American Control
Conference}, 2010, accepted for Publication.

\bibitem{FM86} R.E. Funderlic and C.D. Meyer, \textquotedblleft Sensitivity
of the stationary distribution vector for an ergodic Markov chain," \emph{%
Linear Algebra and Applications}, vol. 76, pp. 1-17, 1986.

\bibitem{Meyer94} C.D. Meyer, \textquotedblleft Sensitivity of the
stationary distribution of a Markov chain," \emph{SIAM J. Matrix Appl.},
vol. 15, pp. 715-728, 1994.
\end{thebibliography}
\end{document}